\documentclass[10pt,letter,final]{IEEEtran}
\usepackage{amsthm,amssymb,graphicx,multirow,amsmath,color,amsfonts}
\usepackage[update,prepend]{epstopdf}%
\usepackage[noadjust]{cite}
\usepackage{tikz}
\usetikzlibrary{arrows,calc}		
\usepackage{bbm} 
\usepackage{pdfpages}
\usepackage{tabulary}
\usepackage{multirow}
\usepackage{comment}
\usepackage{dsfont}
\usepackage{pgf,tikz}
\usepackage{bigints}
\usepackage{mathtools}
\usepackage[english]{babel}
\usepackage[autostyle]{csquotes}
\usepackage{textcomp}
\usepackage{varwidth}

\newcommand{\pgfl}{\mathtt{PGFL}}

\newcommand{\pdf}{\mathtt{pdf}}

\newcommand\numberthis{\addtocounter{equation}{1}\tag{\theequation}}

\newcommand{\x}{\mathbf{x}}
\newcommand{\y}{\mathbf{y}}

\makeatletter
\renewcommand{\env@cases}[1][@{}l@{\quad}l@{}]{%
  \let\@ifnextchar\new@ifnextchar
  \left\lbrace
  \def\arraystretch{1.2}%
  \array{#1}%
}
\makeatother

\def\nb0{{\mathbf{0}}}
\def\nb1{{\mathbf{1}}}

\def\ncalB{{\mathcal{B}}}

\def\ncalE{{\mathcal{E}}}

\def\ncalR{{\mathcal{R}}}

\def\ncalZ{{\mathcal{Z}}}

\def\nbbE{{\mathbb{E}}}

\def\nbbP{{\mathbb{P}}}

\def\nbbR{{\mathbb{R}}}

\newtheorem{lemma}{Lemma}

\newtheorem{definition}{Definition}

\newtheorem{theorem}{Theorem}

\newtheorem{cor}{Corollary}

\def\argmax{\operatorname{arg~max}}
\def\E{\mathbb{E}}

\def\P{\mathbb{P}}

\def\p{p}

\def\sir{\mathtt{SIR}}

\allowdisplaybreaks 

\begin{document}
\graphicspath{{./Figures/}}
\title{Meta Distribution for Downlink NOMA in Cellular Networks with 3GPP-inspired User Ranking}
\author{
Praful D. Mankar and Harpreet S. Dhillon
\thanks{The authors are with Wireless@VT, Bradley Department of Electrical and Computer Engineering, Virginia Tech, Blacksburg, VA. Email: \{prafuldm, hdhillon\}@vt.edu. This work was supported by the US National Science Foundation (NSF)
 under Grant CNS-1814477.
}
}

\maketitle
\thispagestyle{empty}
\pagestyle{empty}
\begin{abstract} 
This paper presents the  meta distribution analysis of the downlink two-user non-orthogonal multiple access (NOMA) in cellular networks. We propose a novel user ranking technique wherein the users from the cell center (CC) and cell edge (CE) regions are paired for the non-orthogonal transmission.
Inspired by how users are partitioned in 3GPP cellular models, the  CC and CE users are characterized based on the mean powers received from the serving and the dominant interfering BSs.  We demonstrate that the  proposed technique ranks users in an accurate order with distinct link qualities, which is imperative for the performance of NOMA system. The exact moments of the meta distributions for the CC and CE users under NOMA and orthogonal multiple access (OMA) are derived. In addition, we provide  tight beta distribution approximations for the meta distributions and exact expressions of the mean local delays and the cell throughputs for the NOMA and OMA cases. To the best of our knowledge, this is the first comprehensive analysis of NOMA using stochastic geometry with 3GPP-inspired user ranking scheme that depends upon {\em both} of the link qualities from the serving and dominant interfering BSs.
\end{abstract}
\section{Introduction}
\label{sec:Introduction}

NOMA technique has received significant attention recently in the context of  5G cellular networks which, unlike the traditional OMA techniques, enables the BSs to serve more than one user using the same resource block (RB); see \cite{ding2017survey} and the references therein.  In NOMA, the transmitter superimposes multiple layers of messages at different power levels and the receiver decodes its intended message using successive interference cancellation  (SIC) technique. A given user first decodes and cancels the interference power resulting from the layers assigned to the users with weaker channel states using SIC and then decodes its own message. 
On the other hand, in OMA, generally the users with poor channel conditions consume most of the RBs in order to meet a certain level of quality of service which lowers the overall spectral efficiency of the system. However, the NOMA technique can meet the quality of service requirements for the users with poor channel conditions without lowering the spectral efficiency of the system by concurrently serving users with poor and better channel conditions using the same spectral resources.

NOMA is configured by ranking the users based on their link qualities which are characterized by path-losses, fading gains and inter-cell interference powers. 
However, incorporating user ranking techniques that depend on all the above components in the stochastic geometry-based system level analysis of downlink NOMA is challenging because of the correlation in the inter-cell interference powers received by the users in a given cell. 
Therefore, most of the existing works in this direction ignore this correlation and instead rank users in the order of their mean signal powers (i.e., link distances) so that the $i$-th closest user becomes the $i$-th strongest user. The set of users scheduled for the non-orthogonal transmission is usually termed as the {\em user cluster}.
The authors of \cite{ali2019downlink,ali2019downlinkLetter,salehi2018meta,salehi2018accuracy} analyzed $N$-ranked NOMA in cellular networks modeled using a Poisson point process (PPP). In \cite{ali2019downlink}, the downlink success probability is derived while forming the user cluster within the indisk of the Poisson-Voronoi (PV) cell. However, the resulting performance estimate may not be truly representative of the NOMA performance gains because users within the indisk of a PV cell will usually experience similar channel conditions and hence lack channel gain imbalance that results in the NOMA gains (see \cite{ding2016impact}).   
 The moments of the {\em meta distribution}, defined in \cite{Martin2016Meta} as the distribution of the successful transmission probability  of the typical link conditioned on the locations of BSs,  are derived for the downlink NOMA in \cite{salehi2018meta,ali2019downlinkLetter} and uplink NOMA in \cite{salehi2018meta}  by ranking users based on their link distances. However, \cite{salehi2018meta} ignores the joint decoding of the subset of layers associated with SIC.  
Nonetheless, assuming this distance-based ranking technique,  \cite{ali2019downlinkLetter,salehi2018meta,salehi2018accuracy} derived the ordered distance distributions of the clustered users while assuming that their link distances follow the distance distribution of the typical link (in the network) independently of each other.
As implied above already, this ignores the fact that the user location in a PV cell is a function of the BS point process. A key unintended consequence of this approach is that it does not necessarily confine the user cluster in a PV cell, which is a significant approximation of the underlying setup  (see Fig. \ref{fig:illustration}, Middle and Right). 
The spectral efficiency of $K$-tier heterogeneous cellular networks is analyzed in \cite{liu2017non} wherein the smaller BSs serve their users using two-user NOMA with the distance-based ranking. 
Besides, \cite{zhang2017downlink} derives the outage probability for the downlink two-user NOMA cellular networks, modeled as a PPP, by ranking the users based on the channel gains normalized  using their received inter-cell  interference powers. Therein, the normalized gains are assumed to be independent and identically distributed (i.i.d.) which again ignores the fact that the link distances and the inter-cell interference powers associated with the users within the same PV cell are correlated.
\begin{figure*}[h]
 \centering
  \hspace{-2cm} \includegraphics[trim=.85cm 1cm .85cm .5cm, width=.35\textwidth]{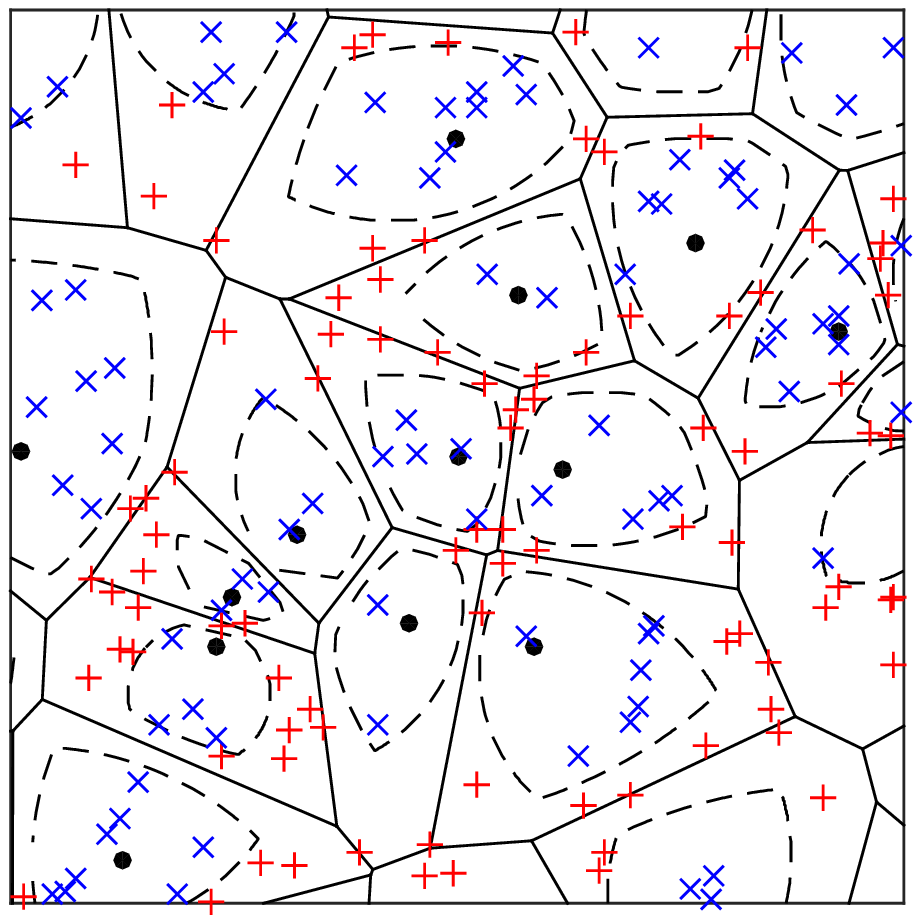}\hspace{-.5cm} 
\includegraphics[trim=.85cm 1cm .85cm .5cm, width=.35\textwidth]{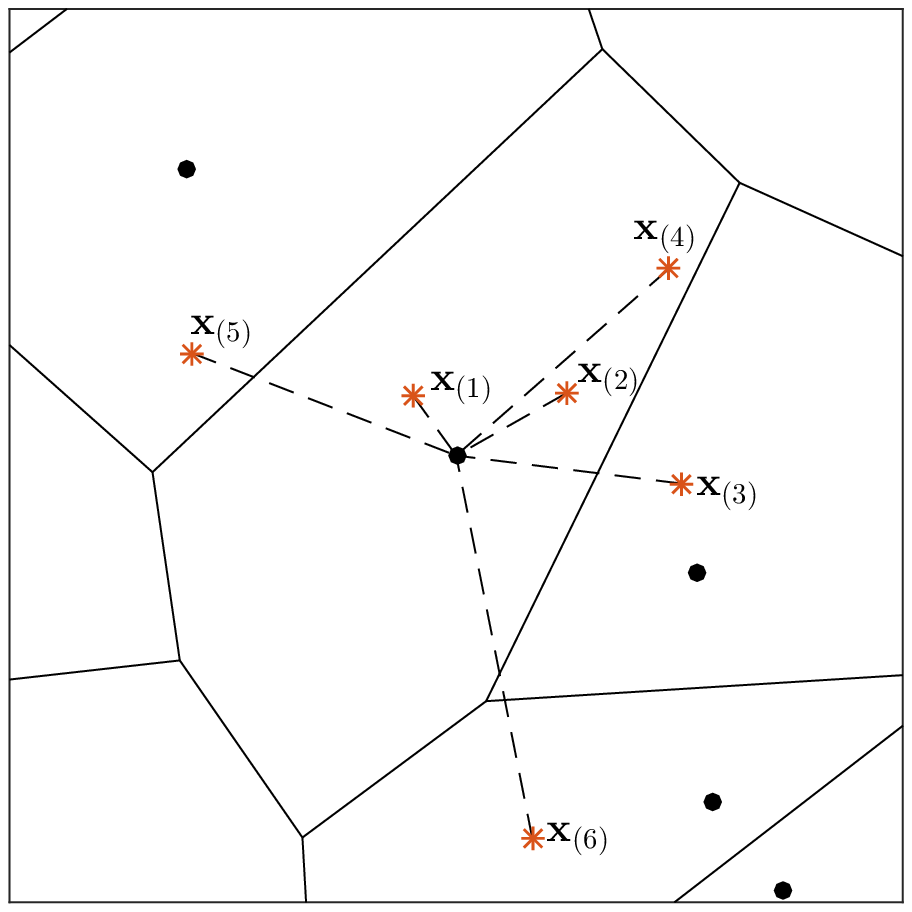}\hspace{-.1cm} 
\includegraphics[trim=.85cm 1cm .85cm .5cm, width=.35\textwidth]{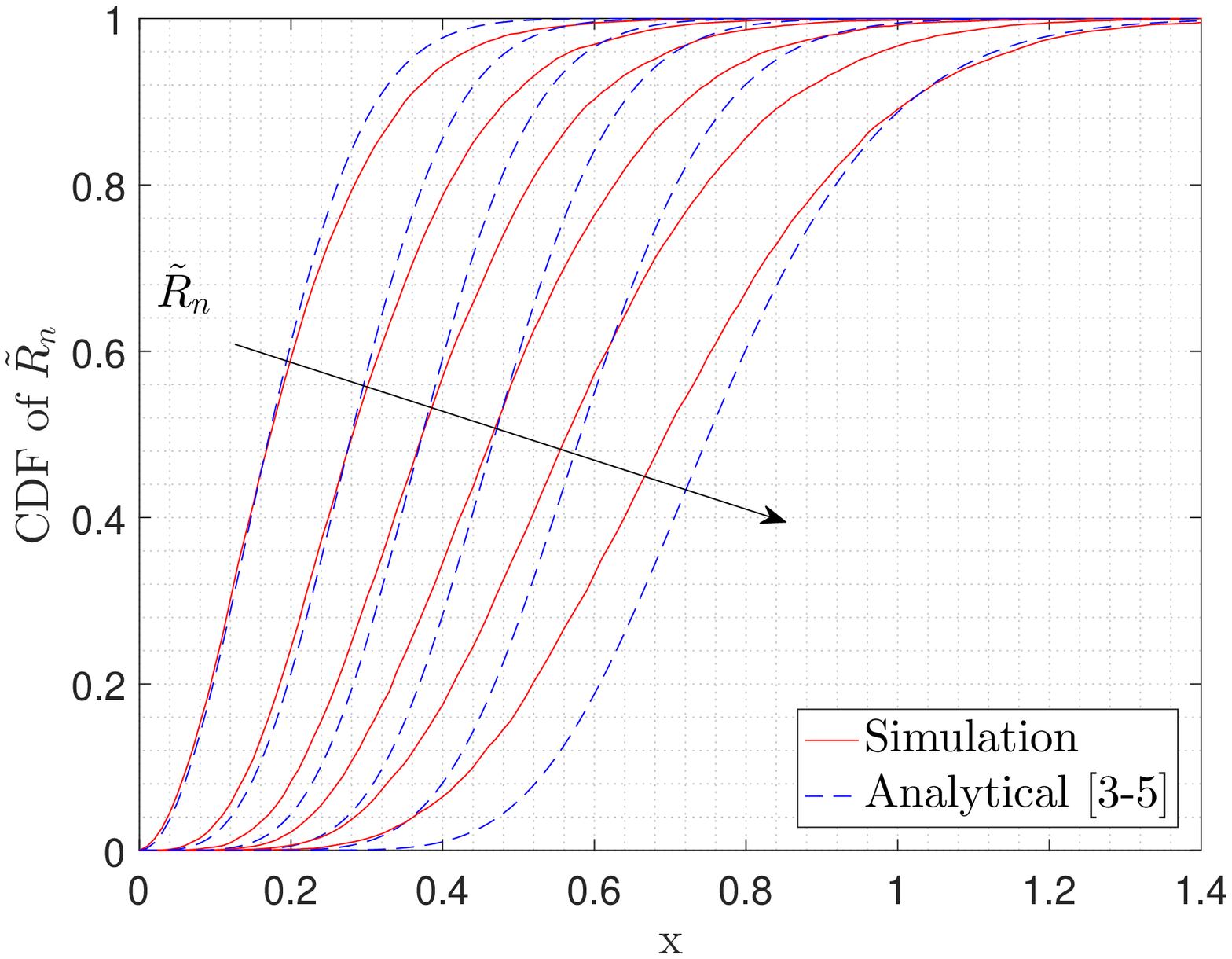}
\caption{Left: an illustration of the classification of the CC and CE users given in \eqref{eq:UserCategorization} for $\tau=0.7$. Middle:  an illustration of the user cluster from  \cite{ali2019downlinkLetter,salehi2018meta,salehi2018accuracy} for $N=6$ (and of the fact that it is not necessarily confined within the PV cell). Right: the distributions of the ordered link distances modeled in \cite{ali2019downlinkLetter,salehi2018meta,salehi2018accuracy} for $N=6$.
 The dot, cross, plus, and star marks correspond to BSs, CC users, CE users, and user cluster, respectively. The solid lines in the left and middle figures correspond to the PV cell boundaries, and the dashed lines in the left figure correspond to the boundaries between CC and CE regions.}
\label{fig:illustration}
\end{figure*}


A more reasonable way of accurately ranking the users is to form the user cluster by selecting users from distinct regions (in order to ensure distinct link qualities for the co-scheduled users). These regions can be constructed based on the ratio of the mean powers (i.e., path-losses) received from the serving and dominant interfering BSs. For instance, the PV cell can be divided into the center (CC) region, wherein the ratio is above a threshold $\tau$, and the cell edge (CE) region, wherein the ratio is below $\tau$. 
A similar approach of classifying users as the CC and CE users is also used in 3GPP LTE to study schemes such as soft frequency reuse (SFR) \cite{dominique2010self}. Inspired by this, we characterize the CC and CE users based on their path-losses from the serving and dominant interfering BSs to pair them for the two-user NOMA system. This way of user pairing is meaningful because of two reasons: 1) order statistic of received signals is dominated by the  path-losses \cite{wildemeersch2014successive},  and 2) the dominant interfering BS contributes most of the interference power  in the PPP setting \cite{Vishnu2017UAV}. 
%

Based on the above pairing technique, we analyze the meta distribution for the downlink NOMA. We first derive the exact moments of the meta distributions for the typical CC and CE users under NOMA. We also provide  tight beta distribution approximations for the meta distributions of the CC and CE users. In addition, the meta distribution analysis for the CC and CE users under OMA is also presented. 
Our results concretely demonstrate that NOMA based on the proposed user pairing technique results in significantly higher  CE user transmission rate and the cell spectral efficiency compared to  OMA. The OMA analysis can also be directly used to analyze other techniques focused on the  performance improvement of the CE user, such as the SFR.\vspace{-.3cm}
\section{System Model}
\label{sec:SystemModel}
\subsection{Network Modeling and User Classification}
We model the BS and user locations using two independent homogeneous PPPs $\Phi$ and $\Psi$ with densities $\lambda_B$ and $\lambda_U$, respectively.  Without loss of generality, we consider that the typical user of $\Psi$ is located at the origin $o$. 
While assuming the strongest BS association policy, the {\em serving link distance} (i.e. distance between the typical user and its serving BS) is given by $R_o=\|\x_o\|$ where $\x_o=\argmax_{\x\in\Phi}\|\x\|^{-\alpha}$ and $\alpha>2$ is the path-loss exponent. Let $R_d=\|\x_d\|$ be the distance from the typical user to its dominant interfering BS where $\x_d=\argmax_{\x\in\Phi_I}\|\x\|^{-\alpha}$ and $\Phi_I=\Phi\setminus\{\x_o\}$ is the point process of the interfering BSs with respect to the typical user. Now, we classify the typical user as either the CC or the CE user based on its distances (i.e. the path-losses) from the serving and dominant interfering BSs as 
\begin{equation}
\text{User}=\begin{cases}
\text{CC user} & \text{if}~ \frac{R_o}{R_d}\leq\tau,\\
\text{CE user} & \text{otherwise},
\end{cases}
\label{eq:UserCategorization}
\end{equation}
where $\tau$ is the threshold which defines the boundary between the CC and CE regions \cite{Praful2018LoadAware}. 
Fig. \ref{fig:illustration} (Left) illustrates the classification of the CC and CE users given in Eq. \eqref{eq:UserCategorization}. From the illustration, it is clear that the criteria given in Eq. \eqref{eq:UserCategorization} accurately preserves the
CE regions wherein the signal-to-intercell-interference ratio ($\sir$) is expected to be lower. As a comparison, Fig. \ref{fig:illustration} (Middle) illustrates a realization of a user cluster that results from the distance-based ranking technique of \cite{ali2019downlinkLetter,salehi2018meta,salehi2018accuracy}. As is clearly evident from the figure, the user cluster is not confined to the PV cell, which is an unintended consequence of ignoring the correlation in the locations of the clustred users. This can also be verified by comparing the distributions of the ordered distances used in \cite{ali2019downlinkLetter,salehi2018meta,salehi2018accuracy} with those obtained from the simulations. This comparison is given in Fig. \ref{fig:illustration} (Right) wherein $\tilde{R}_n$ represents the link distance of the $n$-th closest user from the BS.   \vspace{-.4cm}
\subsection{NOMA Transmission for CC and CE Users}
We assume non-orthogonal transmissions for the CC and CE users from the same cell. Each BS is assumed to transmit signal superimposed of two layers corresponding to the messages for the CC and CE users. Henceforth, the layers intended for the CC and CE users are referred to as the $\mathtt{L_{C}}$ and  $\mathtt{L_{E}}$ layers, respectively.
The $\mathtt{L_{C}}$ and $\mathtt{L_{E}}$ layers are encoded at power levels of $\theta P$ and $(1-\theta)P$, respectively, where $P$ is the transmission power per RB and $\theta \in(0,1)$. Without loss of generality, we assume $P=1$ (since we ignore thermal noise). Usually, NOMA allocates more power to the weaker user (i.e., CE user) so that it receives smaller intra-cell interference power compared to the desired signal power. 
Hence, the CC user first decodes the $\mathtt{L_{E}}$ layer while treating the power assigned to the $\mathtt{L_{C}}$ layer as interference. After successfully decoding the $\mathtt{L_E}$ layer, the CC user cancels its signal using SIC from the received signal and then decodes the $\mathtt{L_C}$ layer. 
Thus, the $\sir$s of the typical user, when being a CC user, for decoding the $\mathtt{L_E}$ and $\mathtt{L_C}$ layers are given by
\begin{align}
\sir_{ce}&=\frac{h_{\x_o}R_o^{-\alpha}(1-\theta)}{h_{\x_o}R_o^{-\alpha}\theta+I_{\Phi_I}}~
\text{and~} \sir_{cc}&=\frac{h_{\x_o}R_o^{-\alpha}\theta}{I_{\Phi_I}},\nonumber
\end{align}
respectively, where $I_{\Phi_I}=\sum_{\x\in\Phi_I}h_{\x}\|\x\|^{-\alpha}$ and $h_\x$ are the channel fading gains which are i.i.d. and follow  unit mean exponential distribution, i.e., $h_{\x}\sim \exp(1)$. 

On the other hand, the CE user decodes $\mathtt{L_E}$ layer while treating the power assigned to the $\mathtt{L_C}$ layer as interference. Thus, the effective $\sir$ of the typical user, when being a CE user, for decoding the $\mathtt{L_E}$ layer becomes
\begin{equation}
\sir_{ee}=\frac{h_{\x_o}R_o^{-\alpha}(1-\theta)}{h_{\x_o}R_o^{-\alpha}\theta+I_{\Phi_I}}.\nonumber
\end{equation}\vspace{-.75cm}
\subsection{Meta Distribution for the NOMA System}
The success probabilities for the CC and CE users are defined as the probabilities that the typical CC and CE users are able to decode their intended messages. 
While this allows to determine the mean success probability of the typical CC and CE users, it does not provide any information on the disparity in the link performance of the CC and CE users across the network.
That said, the conditional success probabilities can be used to acquire more fine-grained information on the disparity in the link performance of these users. The distribution of the conditional success probability is referred to as the meta distribution \cite{Martin2016Meta}. The meta distribution for the CC/CE user can be used to answer questions like what percentage of the CC/CE users can establish their links with the transmission reliability above predefined threshold for given $\sir$ threshold. Thus, building on the  definition of the meta distribution of the $\sir$ in  \cite{Martin2016Meta}, we define the meta distributions for the CC and CE users  under NOMA as below.
\begin{definition}[Meta distribution]
The meta distribution of the typical CC user's success probability is defined as 
\begin{equation}
\bar{F}_{\text{cc}}(\beta_c,\beta_e;x)=\P[\p_c(\beta_c,\beta_e\mid \Phi)>x],
\end{equation}
and the meta distribution of the typical CE user's success probability is defined as 
\begin{equation}
\bar{F}_{\text{ce}}(\beta_e;x)=\P[\p_e(\beta_e\mid\Phi)>x],
\end{equation}
where $x\in[0,1]$, $\beta_c$ and $\beta_e$ are the $\sir$ thresholds corresponding to the $\mathtt{L_C}$ and $ \mathtt{L_E}$ layers, respectively. Further, 
 $p_c(\beta_c,\beta_e\mid\Phi)=\P[\sir_{cc}\geq \beta_c,\sir_{ce}\geq \beta_e\mid \Phi]$ and $p_e(\beta_e\mid \Phi)=\P[\sir_{ee}\geq \beta_e\mid \Phi]$ are conditional success probabilities of the typical CC and CE users, respectively. 
\end{definition}
Note that the meta distribution is measured for the typical CC/CE user conditioned on its location at the origin.\vspace{-.3cm}
\section{Meta Distribution Analysis for CC and CE users under NOMA and OMA}
The key intermediate step in the meta distribution analysis for the CC user (CE user) is the joint distribution of the serving  link distance $R_o=\|\x_o\|$ and  the interfering BSs' distances $\|\x_i\|$, $\x_i\in\Phi_I$, under the condition of $R_o\leq R_d\tau$ ($R_o>R_d\tau$). For this, we first need to obtain the joint probability density functions ($\pdf$s) of $R_o$ and $R_d$ for the CC  and CE users which are presented in the following lemma.
\begin{lemma}
\label{lemma:DistanceDistribution}
The probabilities of the typical user being the CC and CE users are equal to $\tau^2$ and $1-\tau^2$, respectively.
The joint $\pdf$ of $R_o$ and $R_d$ for the CC user is 
\begin{equation}
f^{\text{cc}}_{R_o,R_d}(r_o,r_d)=\frac{(2\pi\lambda_B)^2}{\tau^2}r_or_d\exp\left(-\pi\lambda_B r_d^2\right),
\label{eq:pdf_RoRd_CC}
\end{equation}
for $r_d\geq \frac{r_o}{\tau}$ and $r_o\geq 0$. The joint $\pdf$ of $R_o$ and $R_d$ for the CE user is 
\begin{equation}
f^{\text{ce}}_{R_o,R_d}(r_o,r_d)=\frac{(2\pi\lambda_B)^2}{1-\tau^2}r_or_d\exp\left(-\pi\lambda_B r_d^2\right),
\label{eq:pdf_RoRd_CE}
\end{equation}
for $\frac{r_o}{\tau}>r_d \geq r_o$ and $r_o\geq 0$.
\end{lemma}
\begin{proof}
The joint $\pdf$ of $R_o$ and $R_d$ for the typical user can be written  as \cite{Dhillon2013NonUni}
\begin{align}
f_{R_o,R_d}(r_o,r_d)&=(2\pi\lambda_B)^2r_or_d\exp(-\pi\lambda_B r_d^2),
\label{eq:pdf_RoRd}
\end{align}
for $r_d\geq r_o\geq 0$. Using Eq. \eqref{eq:pdf_RoRd}, the probability of the typical user being the CC user can be obtained as 
\begin{align}
\P\left[R_o\leq R_d\tau\right]&=(2\pi\lambda_B)^2\int\limits_0^\infty\int\limits_0^{ r_d\tau}r_or_d\exp(-\pi\lambda_B r_d^2){\rm d}r_o{\rm d}r_d\nonumber\\
&=\tau^2.
\label{eq:Probabilit_CC}
\end{align} 
Hence, the probability of the typical user being the CE user becomes $1-\tau^2$.
Thus, using Eqs. \eqref{eq:pdf_RoRd} and \eqref{eq:Probabilit_CC} along with the definiations of the CC and CE user given by Eq. \eqref{eq:UserCategorization}, we obtain the final expressions given in Eqs. \eqref{eq:pdf_RoRd_CC} and  \eqref{eq:pdf_RoRd_CE}. 
\end{proof}
 In the following subsections, we first derive the moments of the meta distributions for the CC and CE users under the NOMA case which will be later used to analyze the OMA case, derive a tight approximation for the meta distribution, and determine the mean local delay and the cell throughput.\vspace{-.35cm}
\subsection{Meta Distribution for CC Users under NOMA}
Since the CC user needs to jointly decode the $\mathtt{L_C}$ and $\mathtt{L_E}$ layers  for the successful reception, the successful reception event for the CC user is given by
\begin{align}
\ncalE_c&=\{\sir_{cc}\geq\beta_c\}\cap\{\sir_{ce}\geq\beta_e\}\nonumber\\
&=\left\{h_{\x_o}\geq R_o^\alpha I_{\Phi_I}\chi_c\right\},
\label{eq:SuccessEevent_CC}
\end{align}
where $\chi_c=\max\left\{\frac{\beta_c}{\theta},\frac{\beta_e}{1-\theta(1+\beta_e)}\right\}$. It is easy to interpret that the interference due to non-orthogonal transmission reduces the effective transmission power for decoding the $\mathtt{L_E}$ layer from  $(1-\theta)$ to $\min\{\frac{\beta_e}{\beta_c}\theta ,1-\theta(1+\beta_e)\}$,  which decreases the chance of successful transmission. 
Since it is difficult to directly derive the meta distribution \cite{Martin2016Meta}, we derive the $b$-th moment of the meta distribution for the typical CC user in the following theorem.
\begin{theorem}
\label{thm:MetaDisMoment_CC}
The $b$-th moment of the meta distribution for the typical CC user under NOMA is 
\begin{align}
M_b^{\text{cc}}(\chi_c,\tau)&=\frac{1}{1+\tau^2\ncalZ_b^{\text{cc}}(\chi_c,\tau)}\label{eq:MetaDisMoment_CC_Sigma0}
\end{align}
where $\delta=\frac{2}{\alpha}$ and
\begin{equation}
\ncalZ_b^{\text{cc}}(\chi_c,\tau)=\chi_c^{\delta}\int_{\chi_c^{-\delta}\tau^{-2}}^\infty[1-(1+t^{-\frac{1}{\delta}})^{-b}]{\rm d}t.
\label{eq:Zc}
\end{equation}
\end{theorem}
\begin{proof}
The success probability of the typical CC user conditioned on $\Phi$ is 
\begin{align*}
\p_c(\beta_c,\beta_e\mid \Phi)&=\P\left(\ncalE_c\mid\Phi\right)\stackrel{(a)}{=}\prod\limits_{\x\in\Phi_I}\frac{1}{1+R_o^\alpha \chi_c \|\x\|^{-\alpha}},
\end{align*}
where step (a) follows from the independence of the fading gains. 
Hence, the $b$-th moment can be determined as  
\begin{align*}
&M^{\text{cc}}_b(\chi_c,\tau)=\nbbE_\Phi\left[\prod\limits_{\x\in\Phi_I}\frac{1}{(1+R_o^\alpha \chi_c \|\x\|^{-\alpha})^b}\right]\\
\begin{split}
&\stackrel{(a)}{=}\E_{R_o}\exp\left(-\lambda_B\hspace{-.15cm}\int\limits_{\nbbR^2\setminus\ncalB_{o}(\frac{R_o}{\tau})}\hspace{-.15cm}\left[1-(1+R_o^\alpha \chi_c \|\x\|^{-\alpha})^{-b}\right]{\rm d}\x\right)
\end{split}\\
&=\E_{R_o}\exp\left(-\pi\lambda_B R_o^2 \chi_c^{\delta}\hspace{-.1cm}\int\limits_{\frac{\chi_c^{-\delta}}{\tau^{2}}}^\infty\hspace{-.1cm}\left[1-(1+t^{-\frac{1}{\delta}})^{-b}\right]{\rm d}t\right),\numberthis
\label{eq:MetaDisMomentCondRo_CC}
\end{align*}
where step (a) follows by using probability generating functional ($\pgfl$) of the PPP $\Phi_I$ of density $\lambda_B$ outside of disk $\ncalB_o\left({R_o}/{\tau}\right)$ as  all of the interfering BSs for the CC user must be farther than ${R_o}/{\tau}$. Now using Eq. \eqref{eq:pdf_RoRd_CC}, we obtain the marginal $\pdf$ of $R_o$ for the CC user as 
\begin{equation}
f^{{cc}}_{R_o}(r_o)=\frac{2\pi\lambda_B}{\tau^2}r_o\exp\left(-\pi\lambda_B\frac{r_o^2}{\tau^2}\right), ~\text{for}~r_o\geq 0.
\label{eq:pdfRo_CC} 
\end{equation}
Finally, using Eqs. \eqref{eq:MetaDisMomentCondRo_CC} and   \eqref{eq:pdfRo_CC}, we obtain Eq. \eqref{eq:MetaDisMoment_CC_Sigma0}. This completes the proof.
\end{proof}\vspace{-.65cm}
\subsection{Meta Distribution for CE Users under NOMA}
The CE user decodes its message while treating the signal intended for the CC user as interference. Thus, the successful transmission event for the CE user is given by
\begin{align}
\ncalE_e&=\{\sir_{ee}\geq \beta_e\}=\left\{h_{\x_o}\geq R_o^\alpha I_{\Phi_I}\chi_e\right\},
\label{eq:SuccessEevent_CE}
\end{align}
where $\chi_e=\frac{\beta_e}{1-\theta(1+\beta_e)}$. 
In the following theorem, we derive the $b$-th moment of the  meta distribution for the CE user.
\begin{theorem}[Moments for CE user]
\label{thm:MetaDisMoment_CE}
The $b$-th moment of the meta distribution for the typical CE user under NOMA is
\begin{align}
M_b^{\text{ce}}(\chi_e,\tau)=\frac{1}{1-\tau^2}\int_{\tau^2}^1\frac{(1+v^{\frac{1}{\delta}}\chi_e)^{-b}}{(1+v\ncalZ_b^{\text{ce}}(\chi_e,v^{-1}))^2}{\rm d}v,
\label{eq:MetaDisMoment_CE_Sigma0}
\end{align}
\begin{equation}
\text{where~}\ncalZ_b^{\text{ce}}\left(\chi_e,a\right)=\chi_e^{\delta}\int_{\chi_e^{-\delta}a}^\infty[1-(1+t^{-\frac{1}{\delta}})^{-b}]{\rm d}t.
\label{eq:Ze}
\end{equation}
\end{theorem}
\begin{proof}
For given $R_d$, we can write $I_{\Phi_I}=h_{\x_d}R_d^{-\alpha}+I_{\tilde\Phi_I}$ where $I_{\tilde\Phi_I}=\sum_{\x\in\Phi_I\setminus\{\x_d\}}h_{\x}\|\x\|^{-\alpha}$. Therefore, the success probability of the typical CE user conditioned on $\Phi$ is
\begin{align*}
\p_e(\beta_e\mid\Phi)&=\P[\ncalE_e\mid\Phi]\\
&=\nbbP\left[h_{\x_o}>R_o^\alpha \chi_e(h_{\x_d}R_d^{-\alpha}+I_{\tilde{\Phi}_I})\right]\\
&\stackrel{(a)}{=}\frac{1}{1+R_o^\alpha R_d^{-\alpha}\chi_e}\prod\limits_{\x\in\tilde{\Phi}_I}\frac{1}{1+R_o^\alpha \chi_e \|\x\|^{-\alpha}},
\end{align*}
where step (a) follows from the independence of the channel fading gains. Hence, the $b$-th moment can be determined as $M_b^{\text{ce}}(\chi_e,\tau)$
\begin{align*}
&=\nbbE_\Phi\left[\frac{1}{(1+R_o^\alpha R_d^{-\alpha}\chi_e)^{b}}\prod\limits_{\x\in\tilde{\Phi}_I}\frac{1}{(1+R_o^\alpha \chi_e \|\x\|^{-\alpha})^b}\right]\\
&=\nbbE_{R_o,R_d}\left[\frac{1}{(1+R_o^\alpha R_d^{-\alpha}\chi_e)^b}\times\vphantom{\left[\prod\limits_{\y_i\in\tilde{\Phi}}\frac{1}{1+R_o^\alpha \chi_e \|\x_i\|^{-\alpha}}\mid R_o,R_1\right]}\right.\\
&\left.~~~~~~\nbbE_{\tilde{\Phi}_I}\left[\prod\limits_{\x\in\tilde{\Phi}}\frac{1}{(1+R_o^\alpha \chi_e \|\x\|^{-\alpha})^b}\mid R_o,R_1\right]\right]\\
&\stackrel{(a)}{=}\nbbE_{R_o,R_d}\left[\frac{1}{(1+R_o^\alpha R_d^{-\alpha}\chi_e)^b}\times\vphantom{\int_{\nbbR^2\setminus\ncalB_{o}(R_d)}}\right.\\
&\left.~~\exp\left(-\lambda_B\int_{\nbbR^2\setminus\ncalB_{o}(R_d)}\left[1-(1+R_o^\alpha \chi_e \|\x\|^{-\alpha})^{-b}\right]{\rm d}\x\right)\right]\\
&\stackrel{(b)}{=}\nbbE_{R_o,R_d}\left[\frac{\exp\left(-\pi\lambda_B R_o^2 \ncalZ_b^{\text{ce}}\left(\chi_e,\frac{R_d^2}{R_o^2}\right)\right)}{(1+R_o^\alpha R_d^{-\alpha}\chi_e)^b}\right]\\
&\stackrel{(c)}{=}\frac{2(\pi\lambda_B)^2}{1-\tau^2}\int_{\tau^2}^1\int_0^\infty\frac{u^3}{(1+v^{\frac{\alpha}{2}}\chi_e)^b}\\
&~~~~~~~~~~~~~~~~~~~~\exp\left(-\pi\lambda_B u^2(1+v\ncalZ_b^{\text{ce}}(\chi_e,v^{-1}))\right){\rm d}u{\rm d}v,
\end{align*}
where step (a) follows by using the $\pgfl$ of the PPP $\tilde{\Phi}_I$ of density $\lambda_B$ outside the disk $\ncalB_o\left(R_d\right)$ as all (other than the dominant) interfering BSs for the CE user must be farther than $R_d$.
Step (b) follows using the Cartesian-to-polar
coordinate conversion such that the term $\ncalZ_b^{ce}(\chi_e,(R_d/R_o)^2)$ is obtained as in Eq. \eqref{eq:Ze}. 
Step (c) follows using the joint $\pdf$ of $R_o$ and $R_1$ given in Eq. \eqref{eq:pdf_RoRd_CE} and the substitutions of $(r_o/r_d)^\alpha=v^{\frac{\alpha}{2}}$ and $r_d=u$. Further algebraic manipulations yield  Eq. \eqref{eq:MetaDisMoment_CE_Sigma0}. This completes the proof.
\end{proof}
The following corollary presents simplified expressions for the bounds on the $b$-th moment derived in Theorem \ref{thm:MetaDisMoment_CE}.
\begin{cor}
\label{cor:Sigma0_CE}
The $b$-th moment of the meta distribution for the typical CE user under NOMA can be bounded as 
\begin{align}
&\frac{1}{1-\tau^2}\int_{\tau^2}^1\frac{(1+v^{\frac{1}{\delta}}\chi_e)^{-b}}{(1+v\ncalZ_b^{\text{ce}}(\chi_e,1))^2}{\rm d}v\leq  M_b^{\text{ce}}(\chi_e,\tau)\nonumber\\
&~~~~~~~~~~~~\leq\frac{1}{1-\tau^2}\int_{\tau^2}^1\frac{(1+v^{\frac{1}{\delta}}\chi_e)^{-b}}{(1+v\ncalZ_b^{\text{ce}}(\chi_e,\tau^{-2}))^2}{\rm d}v
\label{eq:MetaDisMoment_CE_Sigma0_Bounds}
\end{align} 
where $\ncalZ_b^{\text{ce}}(\chi_e,a)$ is given in Eq. \eqref{eq:Ze}.
\end{cor}
\begin{proof}
From Eq. \eqref{eq:Ze}, we can observe that $\ncalZ_b^{\text{ce}}(\chi_e,v^{-1})$ is a positive and non-decreasing function of $v$ for $b>0$, whereas $\ncalZ_b^{\text{ce}}(\chi_e,v^{-1})$ is a negative and non-increasing function of $v$ for $b<0$. 
Therefore, for $\tau^2\leq v\leq 1$ (see Eq \eqref{eq:MetaDisMoment_CE_Sigma0}), we have $$~~~~\ncalZ_b^{\text{ce}}(\chi_e,\tau^{-2})\leq\ncalZ_b^{\text{ce}}(\chi_e,v^{-1})\leq\ncalZ_b^{\text{ce}}(\chi_e,1)~\text{for}~ b>0$$
$$\text{and}~\ncalZ_b^{\text{ce}}(\chi_e,\tau^{-2})\geq\ncalZ_b^{\text{ce}}(\chi_e,v^{-1})\geq\ncalZ_b^{\text{ce}}(\chi_e,1)~\text{for}~b<0.$$ Now, note that Eq. \eqref{eq:MetaDisMoment_CE_Sigma0} is a non-increasing function w.r.t $\ncalZ_b^{\text{ce}}(\chi_e,v^{-1})$ when $b>0$, whereas Eq. \eqref{eq:MetaDisMoment_CE_Sigma0} is a non-decreasing function  w.r.t $\ncalZ_b^{\text{ce}}(\chi_e,v^{-1})$ when $b<0$.  
 Therefore, by replacing $\ncalZ_b^{\text{ce}}(\chi_e,v^{-1})$ with $\ncalZ_b^{\text{ce}}(\chi_e,1)$ and $\ncalZ_b^{\text{ce}}(\chi_e,\tau^{-2})$ in Eq. \eqref{eq:MetaDisMoment_CE_Sigma0}, we obtain the bounds on the $b$-th moment given in Eq. \eqref{eq:MetaDisMoment_CE_Sigma0_Bounds}. This completes the proof. 
\end{proof}\vspace{-.65cm}
\subsection{Meta Distribution for CC and CE users under OMA}
\label{subsec:OrthogonlaTransmission}
In OMA, each BS serves its associated users using orthogonal RBs which means that there is no intra-cell interference.  Thus, OMA provides better success probabilities for the CC and CE users compared to NOMA. However, this reduces the transmission instances, depending on the scheduling type, for the CC and CE users which degrades their transmission rates. The following corollary presents the $b$-th moment of meta distribution for the CC and CE users under OMA.
\begin{cor}[Moments for CC and CE users under OMA]
\label{cor:OrthogonalTransmission}
The $b$-th moment of the meta distribution for the typical CC user under OMA is  
\begin{equation}
\tilde{M}_b^{\text{cc}}(\beta_c,\tau)=M_b^{\text{cc}}(\beta_c,\tau),
\end{equation}
where $M_b^{\text{cc}}(\beta_c,\tau)$ is given by Eq. \eqref{eq:MetaDisMoment_CC_Sigma0}.  The $b$-th moment of the meta distribution for the typical CE user under OMA is 
 \begin{equation}
 \tilde{M}_b^{\text{ce}}(\beta_e,\tau)=M_b^{\text{ce}}(\beta_e,\tau),
 \end{equation} 
 where $M_b^{\text{ce}}(\beta_e,\tau)$ is given Eq. \eqref{eq:MetaDisMoment_CE_Sigma0}. Further, the simplified expressions for the bounds on the $b$-th moment of the typical CE user can be obtained by  setting $\chi_e=\beta_e$ in Eq. \eqref{eq:MetaDisMoment_CE_Sigma0_Bounds}.  
\end{cor} 
\begin{proof}
The success probabilities for the typical CC  and CE users under OMA can be written as
$$~~~~\tilde{\p}_c(\beta_c)=\P\left[h_{\x_o}\geq R_o^\alpha I_{\Phi_I}\beta_c\right]$$
$$\text{and}~\tilde{\p}_e(\beta_e)=\P\left[h_{\x_o}\geq R_o^\alpha I_{\Phi_I}\beta_e\right],$$
respectively, which are equivalent to those in the case of NOMA with $\chi_c=\beta_c$ and $\chi_e=\beta_e$. Thus, the proof directly follows from Theorem \ref{thm:MetaDisMoment_CC}, Theorem \ref{thm:MetaDisMoment_CE} and Corollary \ref{cor:Sigma0_CE}.
\end{proof}\vspace{-.75cm}
\subsection{Beta Approximation}
\label{subsec:BetaApproximation}
Using the Gil-Pelaez's inversion theorem \cite{Gil1951} and the moments derived above, we can obtain the exact meta distributions for the typical CC  and CE users.  However, the evaluation of Gil-Pelaez integral is computationally complex.  Therefore, similar to \cite{Martin2016Meta}, we approximate the meta distribution using the beta distribution by matching the means and variances. Thus, the approximated meta distributions for the CC and CE users under NOMA respectively become
\begin{align}
\hspace{-.25cm}\bar{F}_\text{cc}(\chi_c;x)=I_{x}(\mu_1^{cc},\mu_{2}^{cc}) \text{ and } \bar{F}_\text{ce}(\chi_e;x)=I_{x}(\mu_1^{ce},\mu_2^{ce}), 
\end{align}
where $I_x(\cdot,\cdot)$ is a regularized incomplete beta function, $$\mu_1^{ss}= \frac{M_1^{ss}\mu_2^{ss}}{1-M_1^{ss}}\text{ and } \mu_2^{ss}=\frac{(M_1^{ss}-M_2^{ss})(1-M_1^{ss})}{M_2^{ss}-(M_1^{ss})^2}$$ such that $ss=\text{ce}$ for the CC case and $ss=\text{ce}$ for the CE case.
Similarly, the meta distribution for the CC and CE users under OMA can be approximated using the moments given in Corollary \ref{cor:OrthogonalTransmission}.\vspace{-.5cm}
\subsection{Mean Local Delay}
\label{subsec:MeanLocalDelay}
The first inverse moment of the conditional success probability is the {\em mean local delay} which is nothing but the mean number of transmissions required for a successful delivery of the packet when the transmitter retransmits after each failed transmission \cite{Haenggi2013local}.  
Thus, using Theorem \ref{thm:MetaDisMoment_CC}, Theorem \ref{thm:MetaDisMoment_CE} and Corollary \ref{cor:OrthogonalTransmission}, we present the mean local delays of the CC and CE users in the following corollary.
\begin{cor}[Mean local delay]
The mean local delays of the CC user under NOMA and OMA are
\begin{align}
M_{-1}^{\text{cc}}(\chi_c,\tau)&=\frac{1}{1-\frac{\delta}{1-\delta}\chi_c^{1-\delta}\tau^\alpha}\\
\text{and~~}~\tilde{M}_{-1}^{\text{cc}}(\beta_c,\tau)&=\frac{1}{1-\frac{\delta}{1-\delta}\beta_c^{1-\delta}\tau^\alpha},
\end{align}
respectively. The exact expression and bounds of the mean local delay for the CE user under NOMA can be obtained by setting $b=-1$ in Eqs. \eqref{eq:MetaDisMoment_CE_Sigma0} and  \eqref{eq:MetaDisMoment_CE_Sigma0_Bounds}, respectively.  Similarly, the exact expression and bounds of the mean local delay for the CE user under OMA can be obtained by setting $b=-1$ and $\chi_e=\beta_e$ in Eqs. \eqref{eq:MetaDisMoment_CE_Sigma0} and  \eqref{eq:MetaDisMoment_CE_Sigma0_Bounds}, respectively.  
\end{cor}\vspace{-.5cm}
\subsection{Cell Throughput}
\label{subsec:TransmissionRate}
The transmission rates of the CC and CE users can be determined by using the means of their meta distributions. Therefore, the cell throughput under NOMA becomes $\ncalR_{\text{cell}}(\tau)=$
\begin{align}
\log_2(1+\beta_c)M_1^{\text{cc}}(\chi_c,\tau) + \log_2(1+\beta_e)M_1^{\text{ce}}(\chi_e,\tau).
\end{align}
In addition, due to time sharing of RBs, the cell throughput under OMA becomes
\begin{align}
\tilde{\ncalR}_{\text{cell}}(\tau)&=\rho\log_2(1+\beta_c)\tilde{M}_1^{\text{cc}}(\beta_c,\tau) \nonumber\\
&~+ (1-\rho)\log_2(1+\beta_e)\tilde{M}_1^{\text{ce}}(\beta_e,\tau),
\end{align}
where $\rho$ is the fraction of time the CC user is scheduled.\vspace{-.3cm}
\begin{figure*}[h]
 \centering
 \includegraphics[width=.32\textwidth]{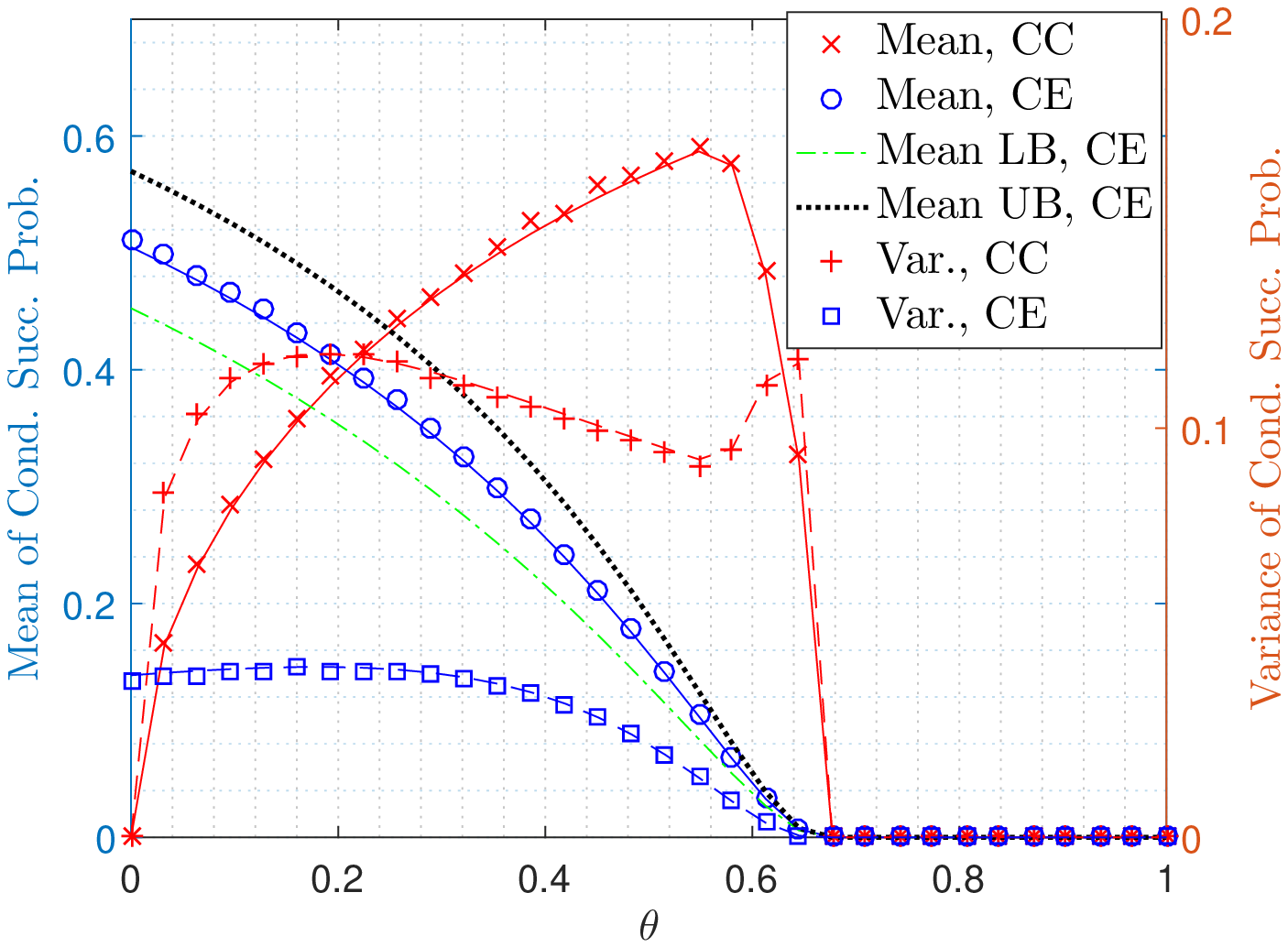}
\includegraphics[width=.32\textwidth]{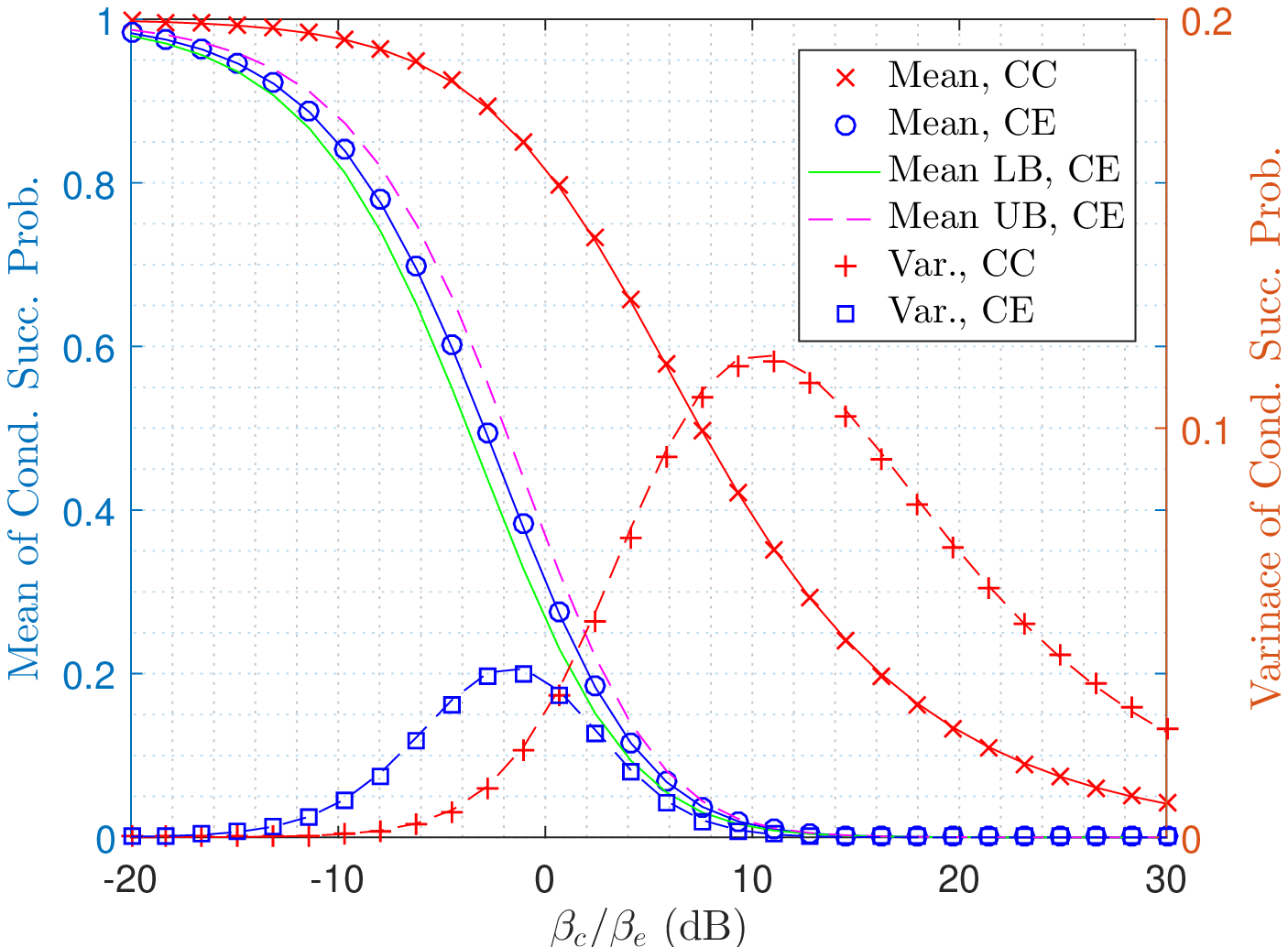}
 \includegraphics[width=.32\textwidth]{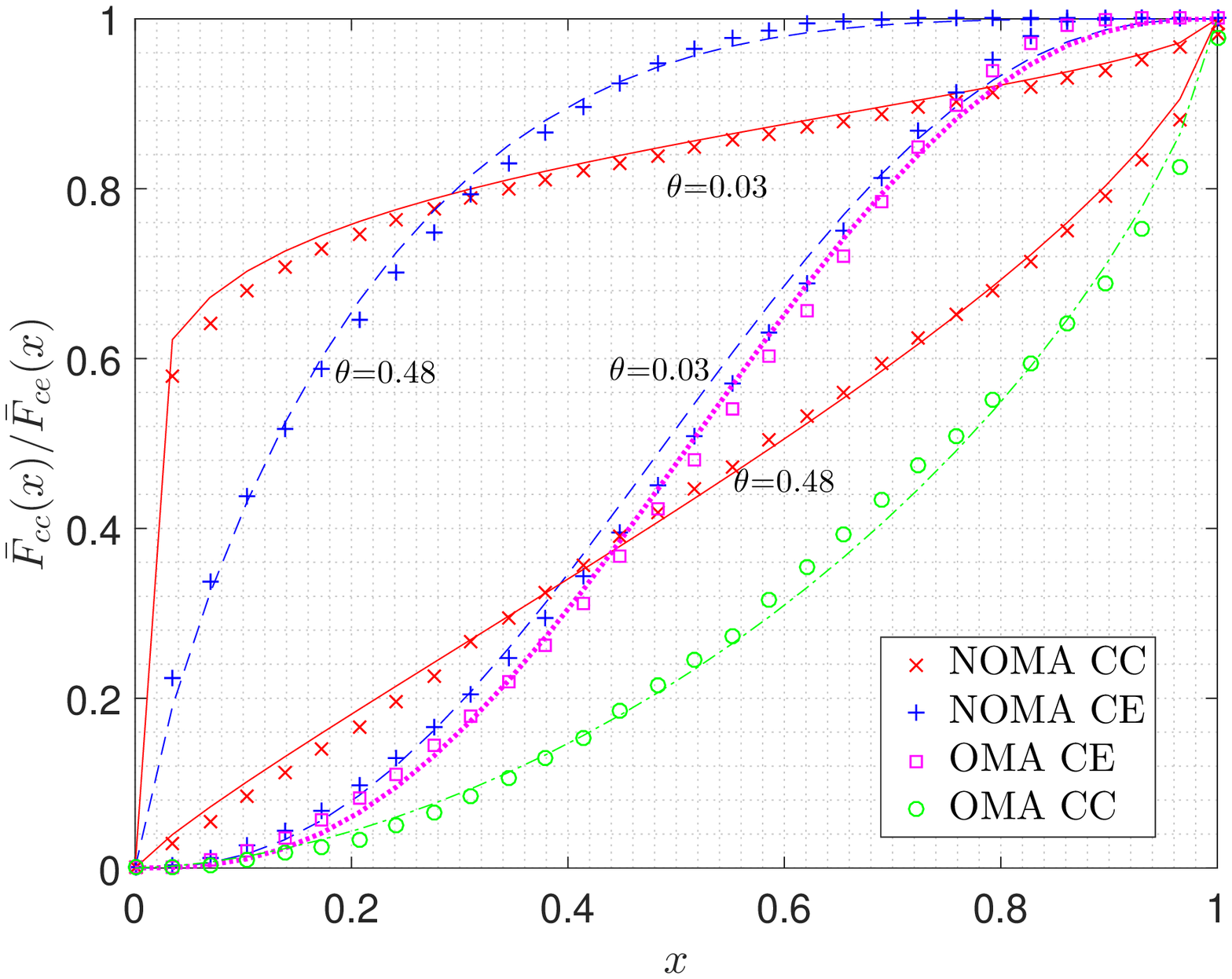}\vspace{-.4cm}
\caption{Moments for the CC and CE users under NOMA (Left) and OMA (Middle). Beta approximation of the meta distribution (Right). LB and UB respectively denotes the lower and upper bounds. The solid and dashed curves correspond to analytical results and markers correspond to simulation results.}\vspace{-.55cm}
\label{fig:Moments_NOMA_OMA_BetaApp}
\end{figure*}
\section{Numerical Results}
 \label{sec:NumericalResults}
In order to verify the analytical results and obtain design insights, we consider the system parameters as $\tau=0.7$, $\alpha=4$, $\lambda_B=1$, $\lambda_U\gg \lambda_B$ (such that each cell can form at least one pair of the CC and CE users) and $(\beta_c,\beta_e)=(3,-3)$ dB, unless mentioned otherwise. Fig. \ref{fig:Moments_NOMA_OMA_BetaApp} (Left) verifies the analysis of  the means and variances of the meta distributions for the CC and CE users under NOMA. 
The moments  for the CE user monotonically decrease with $\theta$ 
since the power allocated to $\mathtt{L_C}$ and $\mathtt{L_E}$ layers negatively affects the success probability for the CE user with increasing $\theta$.
However, the behavior is reversed for the moments of the CC user. This is because while increasing $\theta$ makes it difficult to decode $\mathtt{L_E}$ layer, it also makes it easier to decode $\mathtt{L_C}$ layer at the CC user, which turns out to be the dominant of the two effects in this regime. 
Fig. \ref{fig:Moments_NOMA_OMA_BetaApp} (Middle) verifies the means and variances of the meta distributions for the CC and CE users under OMA. Fig. \ref{fig:Moments_NOMA_OMA_BetaApp} (Left and Middle) also depicts  that the bounds of the mean of the meta distribution (or, the success probability) for the CE user are tight.

Fig. \ref{fig:Moments_NOMA_OMA_BetaApp} (Right) shows that the beta distributions closely approximate the meta distributions for the CC and CE users. Hence, the proposed beta approximations can be used for the system-level analysis of NOMA  without relying on the evaluation of the Gil-Pelaez integrals. 


\begin{figure}[h]
 \centering
 \vspace{-.4cm}
 \includegraphics[width=.45\textwidth]{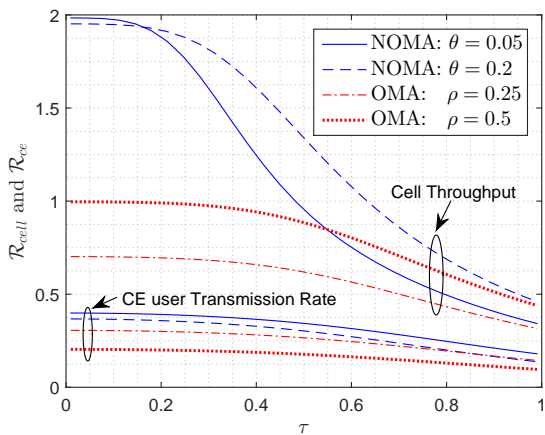}\vspace{-.35cm}
\caption{Cell throughput and CE user's transmission rate.}
\label{fig:CellSumRate_Comparision}\vspace{-.25cm}
\end{figure}
Fig. \ref{fig:CellSumRate_Comparision} shows that both the cell throughput and CE user's transmission rate decrease with $\tau$.
This is because of the success probabilities of both the CC and CE users  degrade with the increase of $\tau$ for given $\theta$ or $\rho$ because of the increase in the inter-cell interference power. We also observe that NOMA can ensure better cell throughput along with improved CE user transmission rate  compared to OMA. It can be seen that the CE user transmission rate increases and the cell throughput decreases as $\theta$ ($\rho$) decreases for NOMA (OMA). Besides, note that decreasing $\theta$ beyond a certain point does not improve the CE user's transmission rate since the success probability of the CE user is limited by the inter-cell interference as $\theta\to 0$.\vspace{-.35cm}
\section{Conclusion}
This paper has provided a comprehensive analysis of downlink two-user NOMA enabled cellular networks. In particular, a new 3GPP-inspired user ranking technique has been proposed wherein the CC and CE users are paired for the non-orthogonal transmission. The CC and CE users are characterized based on the path-losses from the serving and dominant interfering BSs. 
Unlike the ranking techniques used in the literature, the proposed technique ranks users accurately with distinct link qualities  which is important to obtain performance gains in NOMA. The exact expressions have been derived for the moments of the meta distributions for the CC and CE users  under NOMA and OMA. We also provided  tight beta approximations for the meta distributions of the CC and CE users under NOMA and OMA. In addition, we also presented the exact expressions for the mean local delays and the cell throughput. The numerical results demonstrated that NOMA along with the proposed user ranking technique results in a significantly higher cell throughput and CE users' transmission rate compared to OMA.\vspace{-.25cm}
%
%

\end{document}